\documentclass[10pt, conference]{IEEEtran} 

\usepackage{graphicx}
\usepackage{color}
\usepackage{amsmath, amsthm}
\usepackage{amssymb}
\usepackage{epsfig}
\usepackage[font=small,labelsep=none]{caption}
\usepackage{float}

\usepackage{wrapfig}
\usepackage{subfigure}

\newcommand{\F}{\mathbb{F}}

\newcommand{\X}{{\cal X}}

\newcommand{\G}{{\cal G}}
\newcommand{\J}{{\cal J}}
\newcommand{\I}{{\cal I}}

\newcommand{\E}{{\cal E}}

\newcommand{\sgn}{\mbox{sgn}}
\newcommand{\Pf}{\mbox{Pf}}

\newcommand{\tr}{\mbox{tr}}

\newtheorem{Lemma}{{Lemma}}

\newtheorem{Prop}{{Proposition}}

\begin{document}
\bibliographystyle{IEEEtran}

\title{Normal Factor Graphs:\\ A Diagrammatic Approach to Linear Algebra}  

\author{\IEEEauthorblockN{Ali Al-Bashabsheh \hspace{2.5cm} Yongyi Mao}
				\IEEEauthorblockA{School of Information Technology and Engineering \\
													University of Ottawa\\
													Ottawa, Canada\\
													\{aalba059, yymao\}@site.uottawa.ca}
				\and
				\IEEEauthorblockN{Pascal O.~Vontobel}
        \IEEEauthorblockA{Hewlett--Packard Laboratories\\
                         1501 Page Mill Road\\
                         Palo Alto, CA 94304, USA\\
                         pascal.vontobel@ieee.org}
}

\date{}
\maketitle

\baselineskip 11.5pt

\begin{abstract}
Inspired by some new advances on normal factor graphs (NFGs), we introduce NFGs as a simple and intuitive diagrammatic approach towards encoding some concepts from linear algebra.
We illustrate with examples the workings of such an approach and settle a conjecture of Peterson on the Pfaffian.
\end{abstract}

\section{Introduction}

The application of graphical modeling frameworks and of the computational
tools therein has undoubtedly revolutionized the field of coding and
information theory. In particular, factor-graph (FG) based modeling of codes
\cite{frank:factor} has facilitated efficient message-passing algorithms in
decoding capacity-achieving codes, and normal realizations of codes on graphs
\cite{Forney2001:Normal} have brought to surface an elegant duality property
of codes in a graphical context. It is also remarkable that these advances in
the information theory community have resonated with the recent development in
the computer science and statistical physics communities, where joint efforts
are made towards developing efficient algorithmic solvers for various
intractable computational problems and towards characterizing the hardness of
such problems. For example, Valiant \cite{Valiant2004:Holographic} developed
what he calls ``holographic algorithms'' as polynomial-time solvers for a
number of problem families which were previously unknown to be in P. In a
related work \cite{Chertkov:easy}, Chernyak and Chertkov characterized a rich
family of binary graphical models on planar graphs for which the computation
of the partition function is easy. Valiant's holographic algorithms have also
recently found applications in the information theory community for solving
certain constrained-coding capacity problems \cite{Schwartz:2008}.

Inspired by the holographic algorithms and making a joint effort with Forney \cite{Forney:MacWilliams2}, the authors of \cite{Bashabsheh:HolTrans} have recently re-formalized the notion of normal factor graphs (NFGs) by introducing the ``exterior function'' semantics, which unifies Forney's duality result in normal realizations \cite{Forney2001:Normal} and the so-called ``holant theorem'' used by Valiant \cite{Valiant2004:Holographic} to establish the holographic algorithms. This has ignited, most recently, the work of Forney and Vontobel \cite{Forney2011:PartitionFunction}, which pushes further the power of NFGs as an elegant graphical model.

In this paper, we point out yet another direction that NFGs may demonstrate their advantages, namely, their use as a linear algebraic tool. Much of this work is motivated by the notion of ``trace diagrams,'' which has been recently recognized in the mathematics community as an elegant and intuitive ``diagrammatic approach'' towards expressing notions in linear algebra \cite{Predrag:Birdtracks, Penrose:NegativeTensors, Penrose:Road, 
 Peterson:Unshakling, Peterson2009:1}. Our main objective is to demonstrate how NFGs
can be used to provide diagrammatic proofs of some results, and also to show their ability to generalize trace diagrams. Many of the identities used as examples are reproductions of results from \cite{Peterson2009:1} in the language of NFGs and we make no claim of the originality of  these results. However, we believe that the NFG approach is more general than trace diagrams and may
 potentially provide a wider range of algebraic tools. 
Indeed, the lucid approach of NFGs enables us to prove a conjecture of Peterson on the Pfaffian \cite{Peterson:Cayley-Hamilton} \cite{Peterson:Unshakling}. 

The remainder of this paper is structured as follows. Section~\ref{sec:Prelim} is primarily a review of the NFG framework under the exterior-function semantics. Section~\ref{sec:Derived} introduces  some additional notions of NFGs which will be useful in deriving the results of this paper. The main results and example identities are presented in Section~\ref{sec:LinAlg}. The paper is concluded in Section~\ref{sec:Conc}.

\section{Preliminaries}
\label{sec:Prelim}
\subsection{Some Notation}

In this subsection, we briefly list some of our notations. All functions have finite domains (which might be different) and assume their values from the same field $\F$. A univariate function is written as $f(x)$, which will also denote the $x$th component of a vector $f$.
That is, $f$ will be used to denote a univariate function or a vector, and the distinction will be clear from the context -- all vectors in this work are column vectors. Similarly, a bivariate function is denoted $A(x_1,x_2)$, which will also be used to denote the $(x_1,x_2)$th element of a matrix $A$. 
For any matrix $A$, we will use $A^{T}$ to denote its transpose, and if $A$ is a square matrix, we use $\det(A)$
and $\tr(A)$ to denote its determinant and trace, respectively. 

For any positive integer $n$, let $N = \{1, \ldots, n\}$ and let $S_{n}$ be the symmetric group on $N$.
A permutation $\sigma \in S_{n}$ such that $\sigma(j) = i_{j}$ for all $j \in N$ will be written as
$
\sigma = \left( \begin{array}{cccc}	1&2&\cdots&n \cr i_1&i_2&\cdots&i_n \end{array}		\right)
$.
Finally, we use $\varepsilon$ to denote the
Levi-Civita symbol, which is defined as
\begin{eqnarray*}
\varepsilon(x_1, \ldots, x_n) \!=\! \left\{ \!\!\!\!
\begin{array}{ll}
\sgn\left(\!\!\!\! \begin{array}{ccc}	1 & \!\!\!\! \cdots & \!\!\!\! n \cr	x_1 & \!\!\!\!\cdots  & \!\!\!\! x_n 		\end{array} \!\!\!\!\right),	& \!\!\! \left(\!\!\!\!\begin{array}{ccc}	1 & \!\!\!\! \cdots & \!\!\!\! n \cr x_1 & \!\!\!\! \cdots & \!\!\!\! x_n 		\end{array} \!\!\!\!	\right) \in S_{n} \cr
0,			& \!\!\! \mbox{otherwise}
\end{array} \right.
\end{eqnarray*}
for all $(x_1, \ldots, x_n) \in N^{n}$.
Note that the domain of $\varepsilon$ is specified by the number of its arguments and hence we need not be explicit about it.

\subsection{Normal Factor Graphs}

Let $V, E,$ and $D$ be disjoint finite sets and $\{\X_i : i \in E \cup D\}$ be a collection of finite alphabets. 
Further, for any $\I \subseteq E \cup D$, we use $x_{\I}$ to denote the variables set $\left\{x_{i} : i \in \I \right\}$ where $x_{i} \in {\cal X}_i$ for all $i \in \I$.
A \emph{normal factor graph} (NFG) 
\cite{Forney2001:Normal} is a quadruple
$\Omega = (V, E, D, f_{V} )$ defined as follows.
\begin{itemize}
\item $(V, E)$ is a graph with vertex set $V$ and edge set $E$ where an edge is incident on exactly two vertices.
\item $D$ is a set (possibly empty) of \emph{dangling edges} where in contrast to a regular edge, a dangling edge is incident on \emph{exactly one} vertex in $V.$
\item Using $\E(v) = \{e \in E \cup D : e \mbox{ is incident on } v\}$ to denote the set of edges and dangling edges incident on the vertex $v \in V$,  each vertex $v \in V$ is associated with a \emph{local function} $f_{v}$ involving precisely the variables set $x_{\E(v)}$, 
and $f_{V}$ is the collection of all local functions, namely, the set $\{f_{v} : v \in V\}$.
\end{itemize}

Given an NFG $\G = (V,E,D,f_{V})$, we may associate with it the unique form
\[
\sum_{x_{\!E}} \prod_{v \in V} f_{v}(x_{\E(v)}).
\]
We refer to this as the \emph{sum-of-products form} \emph{represented} by $\G$. Clearly, the sum-of-products form is a function of $x_{\!D}$, which we refer to as
the \emph{exterior function} \emph{realized} by $\G$ and denote it by $Z_{\G}(x_D)$. It is not hard to see that while a sum-of-products form is represented by
a unique NFG, an exterior function may be realized by an infinite class of NFGs. Finally, we say that two NFGs $\G_1$ and $\G_2$ are \emph{equivalent}, and write 
$\G_1 = \G_2$, if $Z_{\G_1} = Z_{\G_2}$. In some subsequent figures we may equate an NFG $\G$ to a function $f$, which strictly speaking is not correct, but what
we really mean is that $Z_{\G} = f$. Such practice is unlikely to raise confusion and makes some of the analysis more transparent.

A particular sum-of-products form that will be useful in developing some of the tools in this paper is what we refer to as the ``simple'' sum-of-products form.
We define the \emph{simple} sum-of-products form, denoted $\langle \cdot | \cdot \rangle$, as the sum-of-products form involving exactly two functions with the summation being over all common variables.
More explicitly, for any functions $f(x_{\I})$ and $g(x_{\J})$, the simple sum-of-products form is defined as
\[
\langle f | g \rangle = \sum_{x_{\I \cap \J}} f(x_{\I}) g(x_{\J}).
\]

Fig.~\ref{fig:SimpleSumProdF} shows the NFG $\G$ representing $\langle f | g \rangle$. In this figure each set of variables is treated as a single variable and hence represented by a single edge. If such a set is empty, the corresponding edge (or dangling edge) simply disappears from the figure. With this understanding, the NFG in 
Fig.~\ref{fig:SimpleSumProdF} may encode several notions from linear algebra, namely:
\begin{itemize}
\item $\I \cap \J = \emptyset$, then $\G$ corresponds to the tensor product.	If further, 
	\begin{itemize}
		\item $\I = \emptyset$ and $\J \neq \emptyset$, then the tensor product becomes the scalar-vector product. 
		\item $\I = \J = \emptyset$, then the tensor product becomes the scalar-scalar product. 
	\end{itemize}
\item $\I \cap \J$, $\I \backslash \J$, and $\J \backslash \I$ are all non-empty, then $\G$ corresponds to the matrix-matrix product.
\item $\I \cap \J \neq \emptyset$, $\I \backslash \J \neq \emptyset$, and $\J \backslash \I = \emptyset$, then $\G$ corresponds to the matrix-vector product.
\item $\I = \J \neq \emptyset$ (i.e., $\I \cap \J \neq \emptyset$ and $\I \backslash \J = \J \backslash \I = \emptyset$), then $\G$ corresponds to the vector dot product.
\end{itemize}

\begin{figure}
	\centering \setlength{\unitlength}{0.5cm}
	\begin{picture}(10,1.5)(0,0)
		\put(-4,-2.5){\includegraphics[scale = .55]{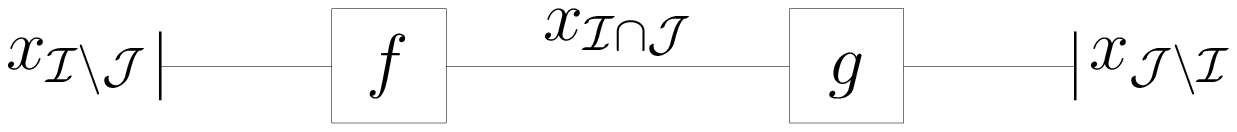} } 
	\end{picture}
	\caption{: The NFG $\G$ representing the simple sum-of-products form
          $\langle f | g \rangle$ where $\I \backslash \J = \{i \in \I : i
          \notin \J \}$ and $\J \backslash \I = \{j \in \J : j \notin \I
          \}$. }
        \vskip-0.25cm
	\label{fig:SimpleSumProdF}
\end{figure}

In this work we find it helpful to adopt the notion of ``ciliation''
\cite{Peterson2009:1} to explicitly indicate the local orderings of variables
at each node. To this end, we add a dot on each node to mark the edge carrying
the local function's first argument and assume the rest of arguments are
encountered in a counter-clockwise manner, {cf.} Fig.~\ref{fig:ciliationF}. We
will not insist on ciliations in every occasion and use them only to
facilitate the analysis.
\vskip0.05cm
\begin{figure}[h]
	\centering \setlength{\unitlength}{0.5cm}
	\begin{picture}(10,5.5)(0,0)
		\put(-3.8,1){\includegraphics[scale = .275]{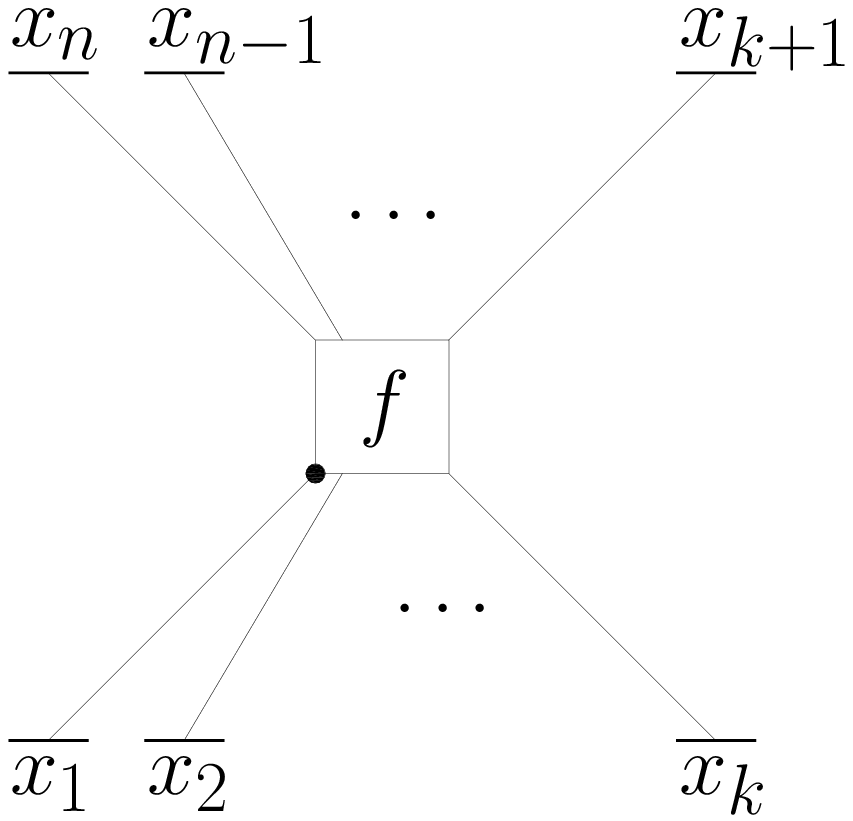} } \put(-1.7,.2){(a)}
		\put(2.4,1){\includegraphics[scale = .275]{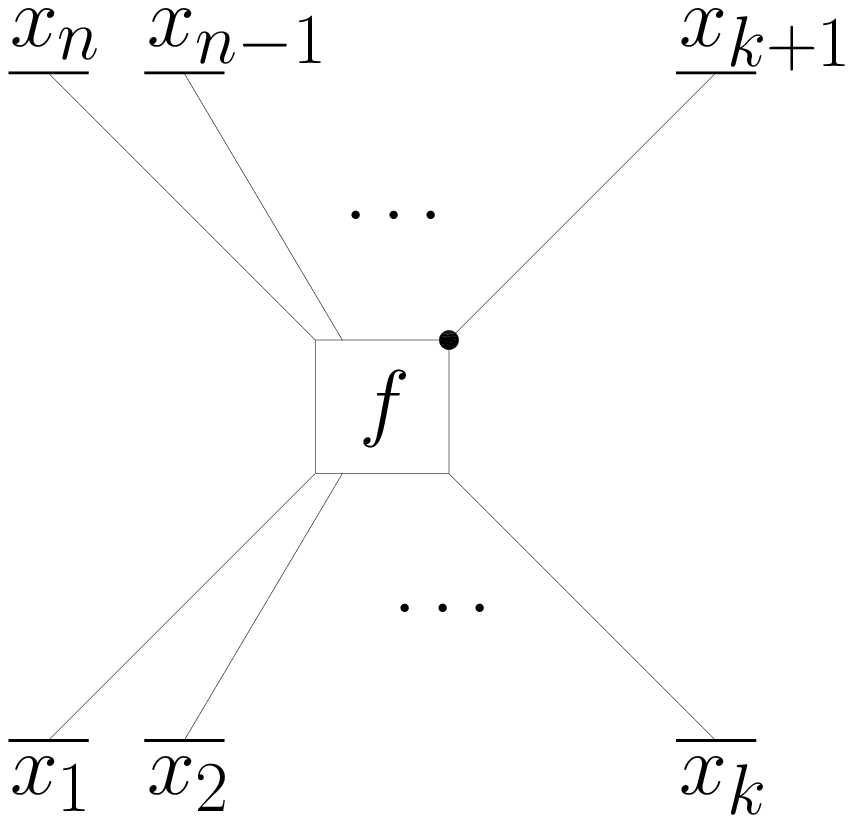} } \put(4.5,.2){(b)}
		\put(8.6,1){\includegraphics[scale = .275]{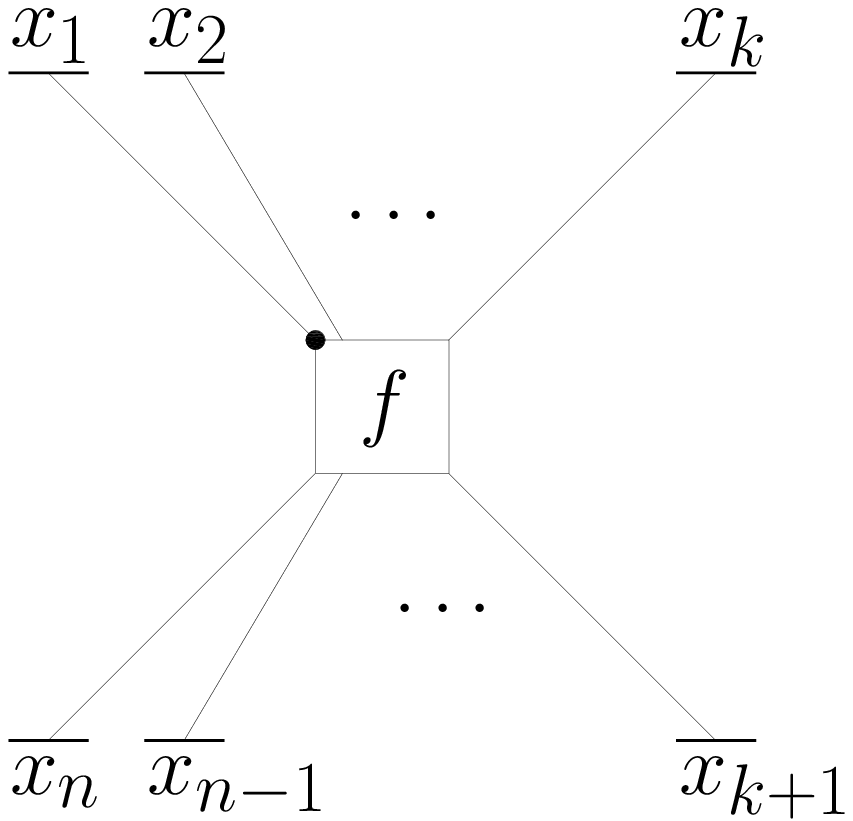} } \put(10.7,.2){(c)}
	\end{picture}
        \vskip-0.1cm
	\caption{: Three assignments of ciliations illustrating different
          orderings of arguments: (a) $f(x_1, \ldots, x_n)$, (b) $f(x_{k+1},
          \ldots, x_n, x_1, \ldots, x_k)$ and (c) $f(x_1, x_n, x_{n-1},
          \ldots, x_2)$. }
        \vskip-0.25cm
	\label{fig:ciliationF}
\end{figure}
\begin{wrapfigure}{r}{0.29\textwidth}
  \begin{center}
    \setlength{\unitlength}{0.25cm}
    \vskip0.75cm
\begin{picture}(10,5.7)(0,0)
		\put(-4.4,1.3){\includegraphics[scale = .35]{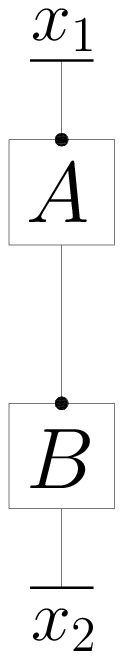} } \put(-3,0){(a)}
		\put(0.6,1.3){\includegraphics[scale = .35]{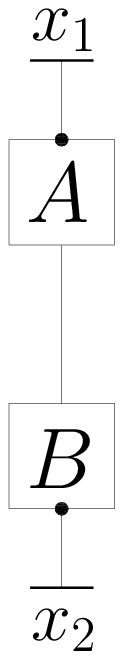} } \put(2,0){(b)}
		\put(5.6,1.3){\includegraphics[scale = .35]{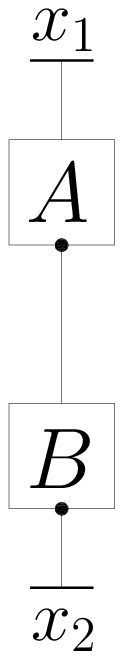} } \put(7,0){(c)}
		\put(10.6,1.3){\includegraphics[scale = .35]{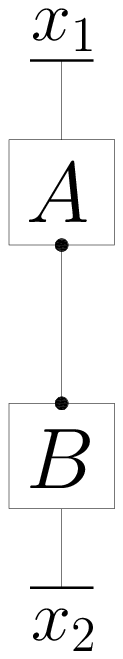} } \put(12,0){(d)} 	
	\end{picture}
  \end{center}
  \caption{: Different arrangements of ciliations: (a) $\G_{1}$, (b)
    $\G_{2}$, (c) $\G_{3}$, and (d) $\G_{4}$.}
  \vskip-0.2cm
  \label{fig:ciliationEX}
\end{wrapfigure}
As an example of the effect of different ciliation arrangements on the realized exterior function, consider the NFGs in Fig.~\ref{fig:ciliationEX} where $A$ and $B$ are $n \times n$ matrices. Then, it is easy to verify that
$ Z_{\G_1} = AB $, $Z_{\G_2} = AB^{T} $, $Z_{\G_3} = A^{T}B^{T} $, and $ Z_{\G_4} = A^{T}B $,
when $Z_{\G_i}$, $i = 1, \ldots, 4$, is viewed as a matrix with rows and columns indexed by $x_1$ and $x_2$, respectively, and the product is the regular matrix product.


\begin{wrapfigure}{r}{0.25\textwidth}
  \vskip-0.25cm
	\centering \setlength{\unitlength}{0.5cm}
        \vskip0.5cm
	\begin{picture}(10,2.75)(0,-1.2)
		\put(0,0){\includegraphics[scale = .55]{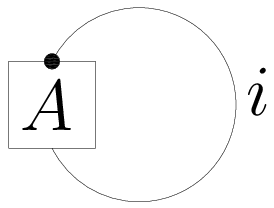} }			\put(1.2,-1){(a)}
		\put(4,0){\includegraphics[scale = .55]{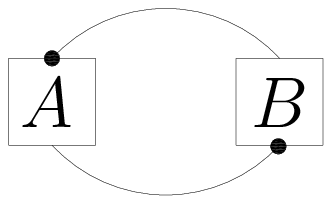} }	\put(5.5,-1){(b)}
	\end{picture}
	\caption{: (a) An NFG that realizes $\tr(A)$ and (b) an NFG that illustrates $\tr(AB) = \tr(BA)$.}
	\label{fig:trace}
\end{wrapfigure}

We end this subsection with an example. Consider the NFG $\G$ in
Fig.~\ref{fig:trace}~(a) where $A$ is an $n \times n$ matrix.  Then from the
definition of the exterior function, we have
\[ Z_{\G} = \sum_{i} A(i,i) = \tr(A). \]
This is the reason behind the name ``trace diagrams'' in \cite{Predrag:Birdtracks, Penrose:NegativeTensors, Penrose:Road}, \cite{Peterson:Unshakling}, \cite{Peterson2009:1}. Briefly, a \emph{trace diagram} is an NFG where each node has degree one, degree two, or is
associated a Levi-Civita symbol. As another example, let $A$ be an $m \times n$ matrix and $B$ be an $n \times m$ matrix. Then
traversing the NFG in Fig.~\ref{fig:trace}~(b) counter-clockwise starting from the upper edge and then starting from the lower edge illustrates the well-known identity, $\tr(AB) = \tr(BA)$. 


\subsection{Vertex Grouping and Vertex Splitting}
\label{sec:grouping_splitting}
Given an NFG $\G$, we may group any two vertex functions of $\G$, say $f$ and $g$, and replace them with a single function node realized by $\langle f | g \rangle$. We refer to this process as \emph{vertex grouping} -- also known as ``closing the box'' \cite{Loeliger:Intro}. Conversely, if $h$ is a vertex node of $\G$ that is realized by $\langle f | g \rangle$ for some $f$ and $g$, then we may split the vertex node $h$ into two function nodes $f$ and $g$. We refer to this as \emph{vertex splitting} -- also known as ``opening the box'' \cite{Loeliger:Intro}. A graphical illustration of vertex grouping and vertex splitting is shown in Fig.~\ref{fig:grouping_splitting}. Note that when we put a dashed box around $f$ and $g$ we mean they are replaced with the single function node $\langle f | g \rangle$, in other words, Figs. \ref{fig:grouping_splitting} (b) and (c) refer to the same NFG.

\begin{figure}[ht]
\setlength{\unitlength}{.5cm}
	\begin{picture}(10,4.5)(-2,0)
		\put(-1.5,1){\includegraphics[scale = .45]{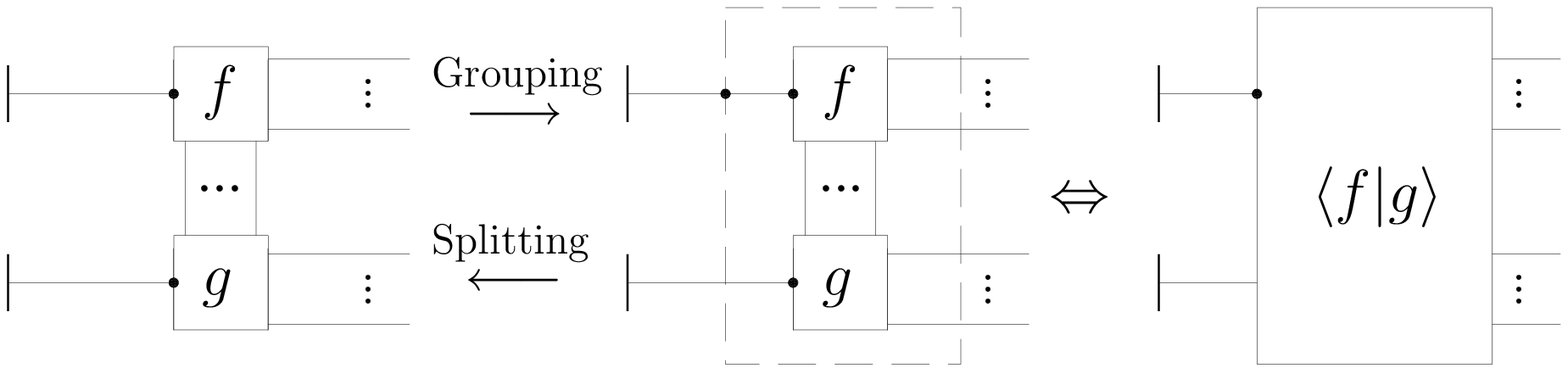} }
		\put(.5,0){(a)}		\put(7.3,0){(b)}		\put(13,0){(c)}
	\end{picture}
	\caption{: Vertex grouping and vertex splitting}
	\label{fig:grouping_splitting}
\end{figure}

Since the range of all the functions in an NFG is a field, the following lemma, which will be the main proving tool to most of the subsequent results, is clear%
\footnote{
We remark that all is needed for Lemma \ref{Lemma:Grouping_Splitting} to hold is the distributive law, and so it is still true even when the range of the local functions of the NFG is a semi-ring rather than a field.
}.
\begin{Lemma}
The exterior function realized by any NFG is invariant under vertex grouping and vertex splitting.
\label{Lemma:Grouping_Splitting}
\end{Lemma}

Clearly one may group an arbitrary number of nodes in an NFG by recursively grouping pairs of vertices; 
ultimately, one may replace an NFG $\G$ with a single function node realizing the exterior function $Z_{\G}$.%

\section{Derived NFGs: Scaling, Addition, and Subtraction of NFGs}
\label{sec:Derived}
We have already seen vertex grouping and splitting as a means of manipulating NFGs. Now, given two NFGs $\G_1$ and $\G_2$ with disjoint dangling edge sets, one of the most basic ways to obtain a new NFG, say $\G$, from $\G_1$ and $\G_2$ is to stack $\G_1$ and $\G_2$ beside each other. Grouping the vertices of $\G_1$ and $\G_2$, it is clear that $Z_{\G}$ is the tensor product of $Z_{\G_1}$ and $Z_{\G_2}$. Of particular interest is the case when $\G_2$ is a single node of degree zero. In this case $Z_{\G_2}$ is a constant, say $\lambda$, and so $Z_{\G} = \lambda Z_{\G_1}$. Moreover, we write $\G = \lambda \G_1$ and say $\G$ is a \emph{scaled} version of $\G_1$ with $\lambda$ being the \emph{scaling} factor%
. 
Graphically, we show scaling by putting the scaling factor beside the NFG.%
\footnote{An equivalent way to obtain $\G$ directly from $\G_1$ is to replace any function node, say $f$, in $\G_1$ with the function node $\lambda f$ in $\G$, and to leave everything else the same.}

Given two NFGs $\G_1$ and $\G_2$ with the same set of dangling edges, we may graphically construct a \emph{compound} NFG, denoted $\G_1 + \G_2$, by drawing the two graphs with the plus sign, ``+", between them, Fig.~\ref{fig:addition}. We use this compound NFG to represent the function $Z_{\G_1} + Z_{\G_2}$, and more formally, we say the compound NFG realizes $Z_{\G_1} + Z_{\G_2}$.
Finally, the ``subtraction of $\G_2$ from $\G_1$" compound NFG, denoted $\G_1 - \G_2$, is defined as $\G_1 + (-1)\G_2$.
\begin{figure}[h]
\centering
\vskip0.25cm
\setlength{\unitlength}{.5cm}
\begin{picture}(5,3.5)(0,0)
\put(-6.5,-.75){\includegraphics[scale = .85]{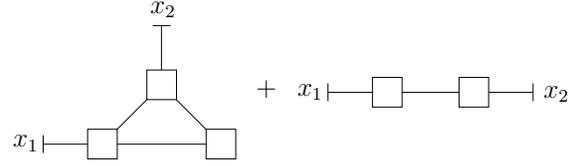}}
\end{picture}
\caption{: Addition of two NFGs.}
\label{fig:addition}
\vskip-0.5cm
\end{figure}

\section{Some Linear Algebra Using NFGs} 
\label{sec:LinAlg}
The main purpose of this section is to demonstrate how NFGs can be used to establish various identities from linear algebra in a simple and intuitive diagrammatic manner.
Throughout the section, we will liberally make use of the following lemma about the Levi-Civita symbol, which easily follows from its definition.
\begin{Lemma}
For any positive integer $n$, it holds that
\[
\varepsilon(x_{1}, x_{2}, \ldots, x_{n-1}, x_{n}) = (-1)^{n-1} \varepsilon(x_{2}, x_{3}, \ldots, x_{n}, x_{1}),
\]
and so, when $n$ is odd, $\varepsilon$ becomes invariant under cyclic-shifts of its arguments.
\label{remark:g}
\end{Lemma}

\subsection{Kronecker Delta}

\begin{wrapfigure}{r}{0.14\textwidth}
\centering\setlength{\unitlength}{.5cm}
  \vskip-1.2cm
	\begin{picture}(10,5)(0,0)
 		\put(-2.1,.75){\includegraphics[scale = .45]{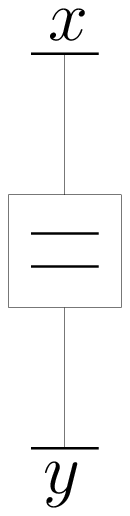}		}			  \put(1.3,0){(a)}
 		\put(0.1,.75){\includegraphics[scale = .45]{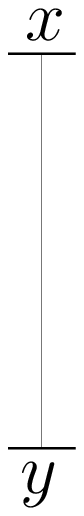}		}			\put(3.5,0){(b)}
	\end{picture}
	\caption{: Both NFGs in (a) and (b) realize $\delta(x-y)$.}
	\label{fig:Delta}
  \vskip-0.7cm
\end{wrapfigure}
The Kronecker delta, denoted $\delta(x)$, is the indicator function evaluating to one if and only if $x = 0$.
An NFG realizing $\delta(x-y)$ is shown in Fig.~\ref{fig:Delta}~(a), which from now on will be abbreviated as the NFG in (b).

\begin{wrapfigure}{r}{0.14\textwidth}
  \centering
  \vskip0.2cm
  \setlength{\unitlength}{.4cm}
  \begin{picture}(10,4.7)(0,0)
   \put(0.5,0){\includegraphics[scale = .4]{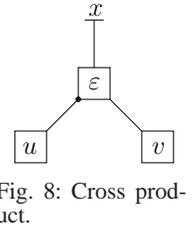}}
  \end{picture}
  \caption{: Cross product.}
  \label{fig:CrossProduct}
\end{wrapfigure}

\subsection{Cross Product}

In this subsection we look at the cross product of vectors of length 3. Let
$\G$ be as in Fig.~\ref{fig:CrossProduct}, then direct computation shows that
\vskip-0.55cm
\begin{align*}
  Z_{\G}(1)
    &= u(2)v(3) - u(3)v(2), \\
  Z_{\G}(2)
    &= u(3)v(1) - u(1)v(3), \\
  Z_{\G}(3)
    &= u(1)v(2) - u(2)v(1).
\end{align*}
\vskip-0.1cm
\noindent I.e., when viewed as vectors in $\F^{3}$, the exterior function
$Z_{\G}$ is the cross product ${u} \times {v}$.

A well-known fact, and straightforward to prove, is the following identity about the contraction of two Levi-Civita symbols, namely \vspace{-.05cm}
\begin{eqnarray*}
\sum_{t}  \varepsilon(y_1,y_2,t) \hspace{-.68cm}&& \varepsilon(t,x_2,x_1) = \\
&& \hspace{-.3cm} \delta(x_1-y_2)\delta(x_2-y_1) - \delta(x_1-y_1)\delta(x_2-y_2).
\end{eqnarray*}
A graphical illustration of the Levi-Civia contraction identity is shown in Fig.~\ref{fig:LC_Contraction}.
\begin{figure}[ht]
  \vskip0.5cm
\centering\setlength{\unitlength}{.5cm}
	\begin{picture}(10,5.75)(1,1.5)
 		\put(-2,1){\includegraphics[scale = .45]{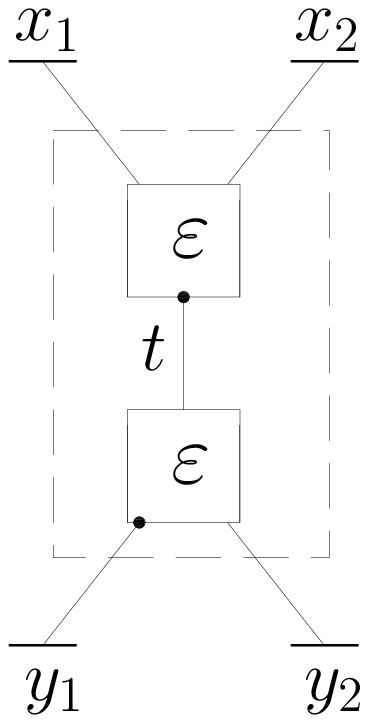}		}			  \put(3,4.5){$=$}
 		\put(4,2.5){\includegraphics[scale = .45]{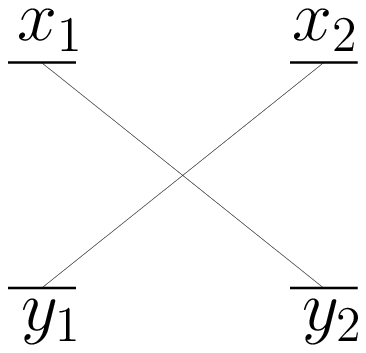}		}			  \put(8.5,4.5){$-$}
 		\put(9,2.5){\includegraphics[scale = .45]{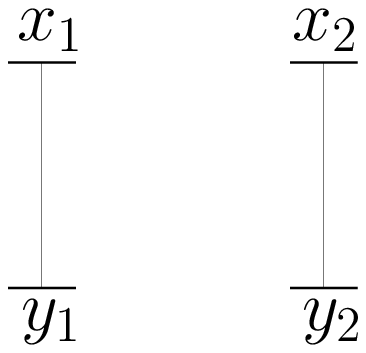}		}		
	\end{picture}
	\caption{: Contraction of two Levi-Civita symbols.}
	\label{fig:LC_Contraction}
\end{figure}

Gluing single variate nodes $w$, $s$, $u$, and $v$, respectively, on the dangling edges $x_1$, $x_2$, $y_1$, and $y_2$ in Fig.~\ref{fig:LC_Contraction} results in Fig.~\ref{fig:CrossProd_Id}. From the previous analysis (in particular, Lemmas~\ref{Lemma:Grouping_Splitting} and \ref{remark:g}), the following identities are then clear: 
\begin{eqnarray}
({u} \times {v}) \cdot ({s} \times {w})
\hspace{-.6cm} &&= \left(({u} \times {v}) \times {s} \right) \cdot {w}
= \left({w} \times ({u} \times {v})  \right) \cdot {s} \nonumber \\
&&= \left(({s} \times {w}) \times {u} \right) \cdot {v}
= \left({v} \times ({s} \times {w})  \right) \cdot u \nonumber \\
&&= ({u} \cdot {s}) ({v} \cdot {w})  - ({u} \cdot {w}) ({v} \cdot {s}). \nonumber
\label{eq:CrossProd_Id}
\end{eqnarray}

\begin{figure}[h]
\centering
\setlength\unitlength{.5cm}
	\begin{picture}(10,7)(2,0)
		\put(-2,1){\includegraphics[scale = .4]{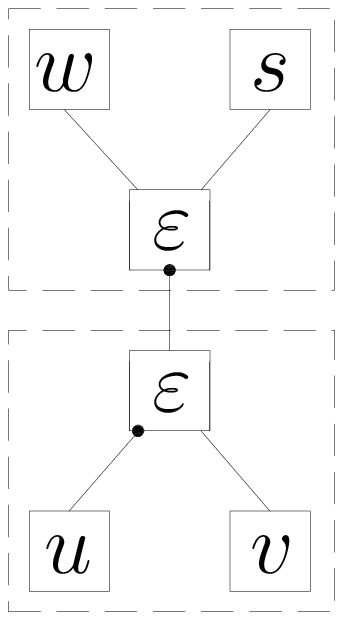}}		\put(1.9,3.7){$=$}
		\put(2.5,1){\includegraphics[scale = .4]{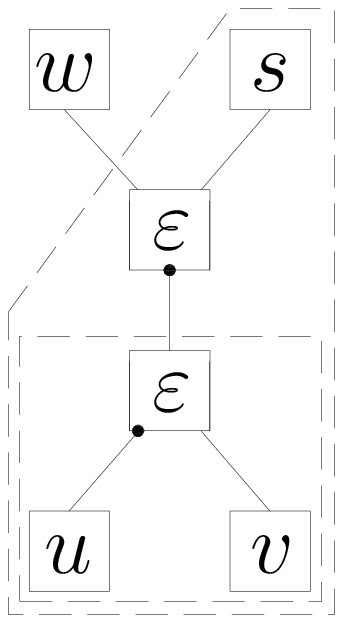}}		\put(6.5,3.7){$=$}
		\put(7,1){\includegraphics[scale = .4]{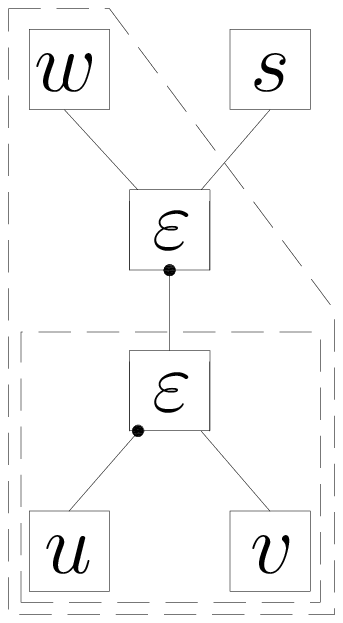}}			\put(10.8,3.7){$=$}
		\put(11.5,1){\includegraphics[scale = .4]{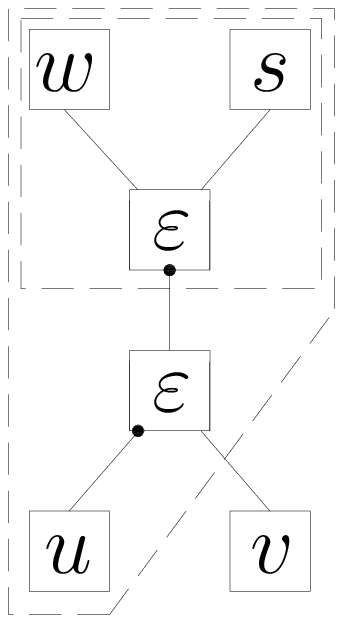}}
	\end{picture}
	\begin{picture}(10,7)(2,0)
		\put(-1.5,4.6){$=$}
		\put(-1.2,2){\includegraphics[scale = .4]{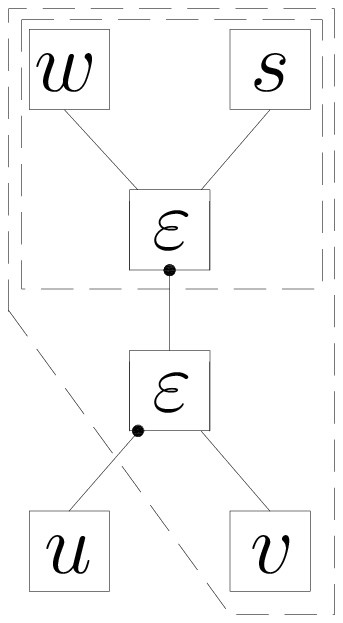}}		\put(2.8,4.6){$=$}
		\put(3,2){\includegraphics[scale = .4]{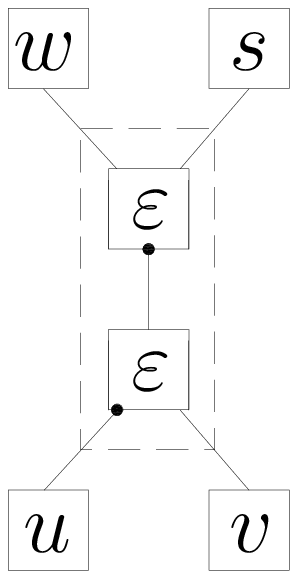}}		\put(6.3,4.6){$=$}
		\put(7,3){\includegraphics[scale = .4]{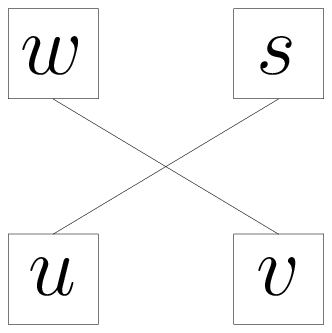}}		\put(10.8,4.6){$-$}
		\put(11.5,3){\includegraphics[scale = .4]{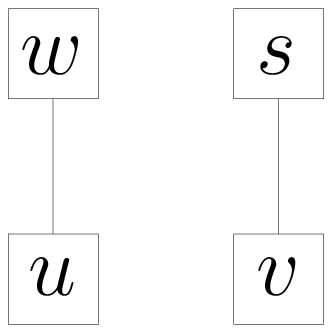}}		
	\end{picture}
        \vskip-1.25cm
	\caption{: A proof of $({u} \times {v}) \cdot ({s} \times {w}) \!=\!
          \cdots \!=\! ({u} \cdot {s}) ({v} \cdot {w})  - ({u} \cdot {w}) ({v}
          \cdot {s})$. }
        \vskip-0.5cm
	\label{fig:CrossProd_Id}
\end{figure}

Many identities regarding the cross product may be obtained in the same way. For instance, let $m_a$, $m_b$, $m_c$, and $m_d$ be arbitrary positive integers
and let $\{a_{i} : 1 \leq i \leq m_a\}$, $\{b_{i} : 1 \leq i \leq m_b\}$, $\{c_{i} : 1 \leq i \leq m_c\}$,
and $\{d_{i} : 1 \leq i \leq m_d\}$ be four collections of length-three vectors. 
Further, let $A$ be the matrix whose $i$th column is $a_{i}$, i.e., $A = (a_{i}: 1\leq i \leq m_a )$, and similarly let $B = (b_{i} : 1 \leq i \leq m_b)$, $C = (c_{i} : 1 \leq i \leq m_c)$,
and $D = (d_{i} : 1 \leq i \leq m_d)$.
We claim that if $m_a = m_d = m$ and $m_b = m_c = m'$ for some $m$ and $m'$, then
\[
\sum_{i=1}^{m}\sum_{j=1}^{m'} (a_i \times b_j) \cdot (c_j \times d_i) \!=\! \tr(AD^{T}BC^{T}) - \tr(BC^{T}) \tr(AD^{T}).
\]
 A detailed proof of this identity is shown in Fig.~\ref{fig:CrossProd_Id_m}.
\begin{figure}
\vskip0.5cm
\centering
\setlength\unitlength{.5cm}
	\begin{picture}(10,6)(2,1.25)
		\put(-2,3.4){$\sum_{i,j} (a_i \times b_j) \! \cdot \! (c_j \times d_i) \! = \! \sum_{i,j}$}
		\put(6.5,0){\includegraphics[scale = .5]{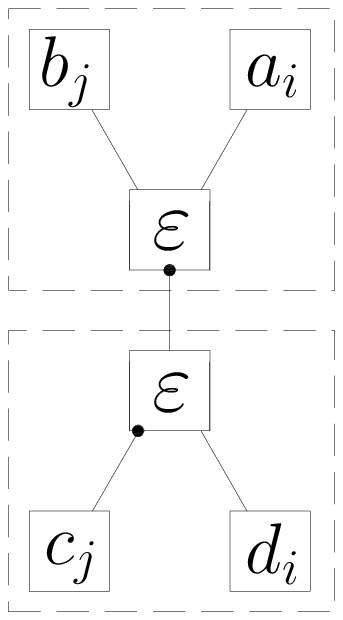}}		\put(10.7,3.4){$= \sum_{i,j}$}
		\put(11.6,0){\includegraphics[scale = .5]{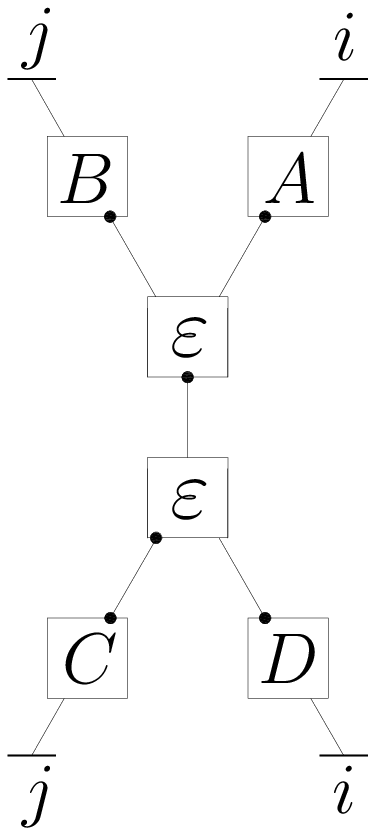}}		
	\end{picture}
	\begin{picture}(10,7)(2,.75)
		\put(-1.5,3.4){$=$}
		\put(-1.2,0){\includegraphics[scale = .5]{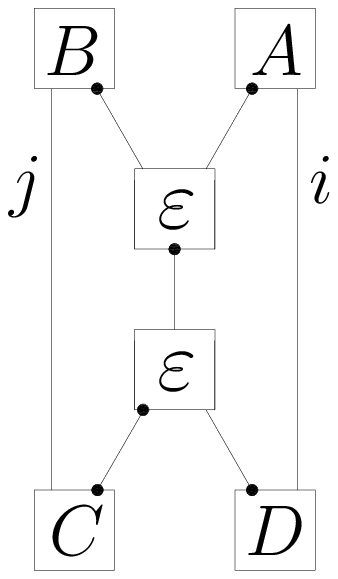}}		\put(3,3.4){$=$}
		\put(4,1.5){\includegraphics[scale = .5]{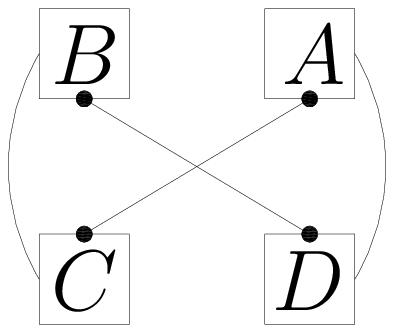}}		\put(8.5,3.4){$-$}
		\put(9,1.5){\includegraphics[scale = .5]{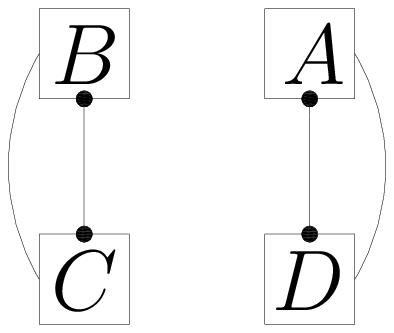}}	
		\put(3,1){$= \ \ \tr(AD^{T}BC^{T}) - \tr(BC^{T}) \tr(AD^{T})$}
	\end{picture}
	\caption{}
	\label{fig:CrossProd_Id_m}
        \vskip0cm
\end{figure}

We emphasize that in practice, one does not need to go through all the steps in Fig.~\ref{fig:CrossProd_Id_m}, i.e., 
one should be able to deduce the identity easily from the last two NFGs.
Slight variations of the NFG before the last equivalence in Fig.~\ref{fig:CrossProd_Id_m} 
may be used to prove 
more identities about the cross product. For instance, when $m_a = m_b = m$ and $m_c = m_d = m'$ for some $m$ and $m'$, 
Fig.~\ref{fig:CrossProd_Id2_m} (a) shows that
\[
\sum_{i=1}^{m}\sum_{j=1}^{m'} (a_i \times b_i) \cdot (c_j \times d_j) = \tr(AB^{T}DC^{T}) - \tr(AB^{T}CD^{T}),
\]
and when $m_{a} = 1$ and $m_b = m_c = m$ for some $m$, Fig.~\ref{fig:CrossProd_Id2_m}~(b) proves
$\sum_{i = 1}^{m} (a_1 \times b_i) \times c_i = \left(BC^{T} \right) a_1 - \tr \left(BC^{T} \right) a_1$.

We invite the reader to verify some of the identities above using traditional methods and contrast with the current diagrammatic approach.

\begin{figure}[h]
\centering
\setlength\unitlength{.5cm}
	\begin{picture}(10,5.5)(2,0)
		\put(-2.5,0){\includegraphics[scale = .45]{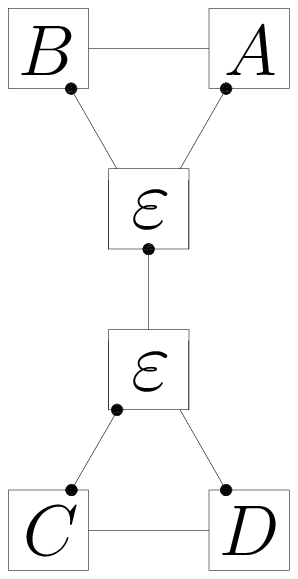}}
		\put(.4,3){$=$} \put(.7,1.7){\includegraphics[scale = .35]{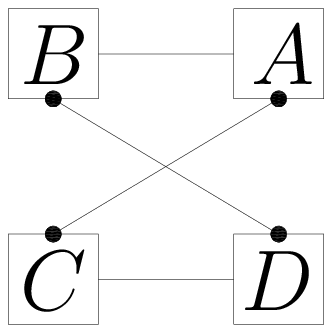}}
		\put(3.3,3){$-$} \put(3.5,1.7){\includegraphics[scale = .35]{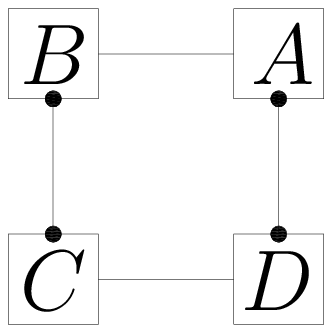}}
		\put(2.7,0){(a)}		
		\put(6.7,0){\includegraphics[scale = .45]{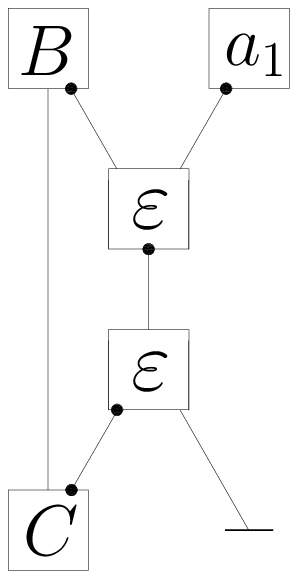}}	
		\put(9.5,3){$=$} \put(10.1,1.7){\includegraphics[scale = .35]{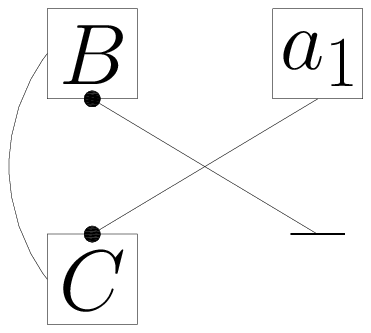}}
		\put(12.6,3){$-$} \put(13.2,1.7){\includegraphics[scale = .35]{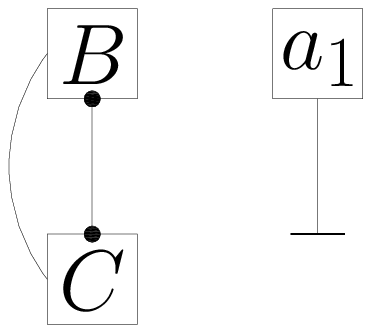}}		
		\put(11,0){(b)}
	\end{picture}
	\caption{}
        \vskip-0.65cm
	\label{fig:CrossProd_Id2_m}
\end{figure}

\subsection{Determinant}
For any $n \times n$ matrix $A$, the determinant is defined as 
\[
\det(A) = \sum_{\sigma \in S_{n}} \sgn(\sigma) \prod_{j = 1}^{n} A(j, \sigma(j)).
\]
From the definitions of exterior function and determinant, it follows that
\begin{figure}[H]
	\centering \setlength{\unitlength}{0.5cm}	
	\begin{picture}(10,2)(0,1.3)
		\put(-3,0){\includegraphics[scale = .45]{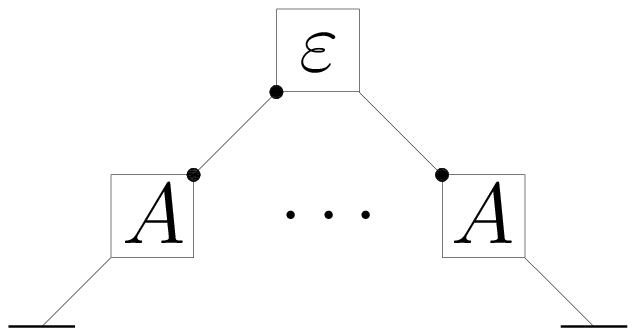}	}			
		\put(4,2){$= \ \ \det(A)$}	
		\put(6.5,0){\includegraphics[scale = .45]{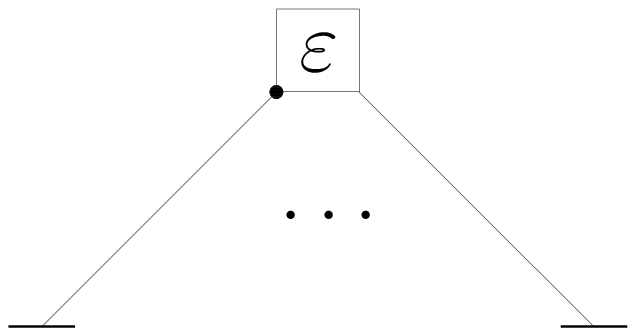}	}		
	\end{picture}
\end{figure}
From this, several properties of the determinant may be obtained, for instance
\begin{figure}[H]
	\centering \setlength{\unitlength}{0.5cm}	
	\begin{picture}(10,2.5)(0,0)
		\put(-4,2){$\det(AB)$}
		\put(-2.5,0){\includegraphics[scale = .37]{fig/Det1_I_fig.ps}	}	 \put(2.6,2){$=$}		
		\put(3,0){\includegraphics[scale = .37]{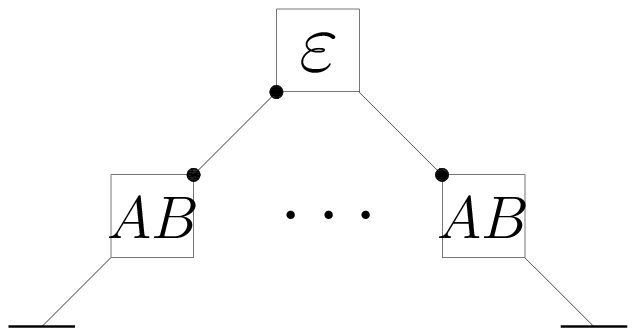}	}  \put(8.2,2){$=$}
		\put(9,0){\includegraphics[scale = .37]{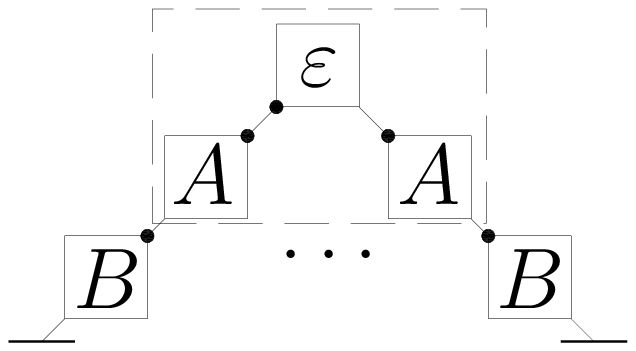}	}
	\end{picture}
	\begin{picture}(10,2)(0,.75)
		\put(-4,2){$= \det(A)$}
		\put(-2.5,0){\includegraphics[scale = .37]{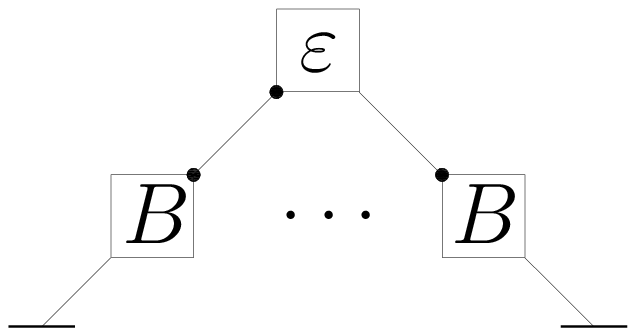}	}		\put(2.6,2){$= \ \det(A)\det(B)$}
		\put(6.2,0){\includegraphics[scale = .37]{fig/Det1_I_fig.ps}	}
	\end{picture}%
\end{figure}
\vskip-0.4cm
\noindent shows that $\det(AB) = \det(A)\det(B)$ for any $n \times
n$ matrices $A$ and $B$.  Another example follows from \vskip0.1cm
\begin{figure}[H]
	\centering \setlength{\unitlength}{0.5cm}	
	\begin{picture}(10,3.5)(0,0)
		\put(-3.8,2){$\det(A^{T})$} \put(-2.2,0){\includegraphics[scale = .37]{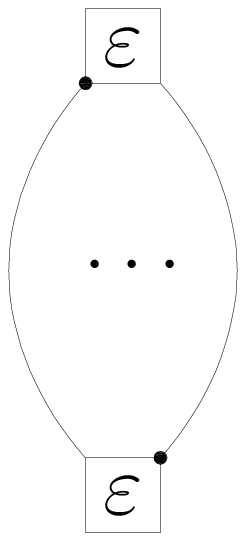}	}	 
		\put(.8,2){$=$} \put(1.2,0){\includegraphics[scale = .37]{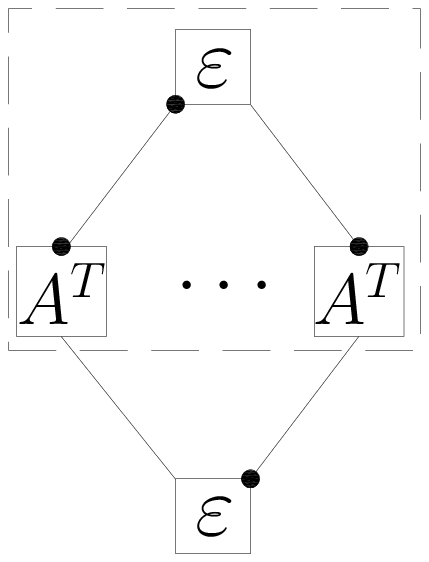}	}  
		\put(4.9,2){$=$} \put(5.4,0){\includegraphics[scale = .37]{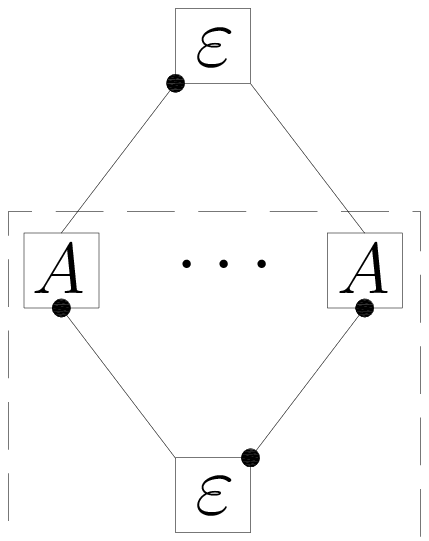}	}
		\put(9.1,2){$= \det(A)$} \put(11,0){\includegraphics[scale = .37]{fig/Det_Transpose_I_fig.ps}	}	 
	\end{picture}
\end{figure}
\vskip-0.45cm
\noindent showing $\det(A^{T}) = \det(A)$.  As a final example, we have
\begin{figure}[H]
\vskip0.25cm
\centering
\setlength\unitlength{.5cm}
\begin{picture}(10,4)(0,2)
	\put(-5.5,0){\includegraphics[scale = .52]{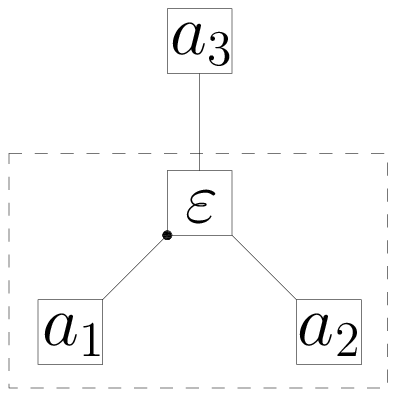}}
	\put(1,3.6){$=$} \put(-1,0){\includegraphics[scale = .52]{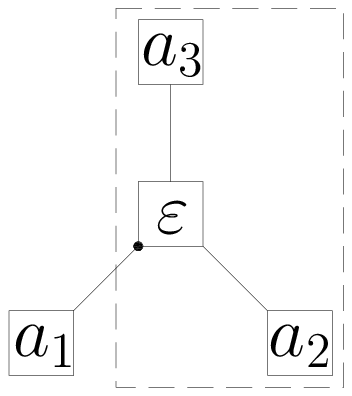}}
	\put(5,3.6){$=$} \put(3.8,0){\includegraphics[scale = .52]{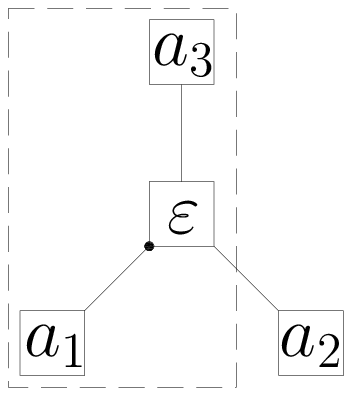}}
	\put(9,3.6){$=$} \put(8.2,.3){\includegraphics[scale = .45]{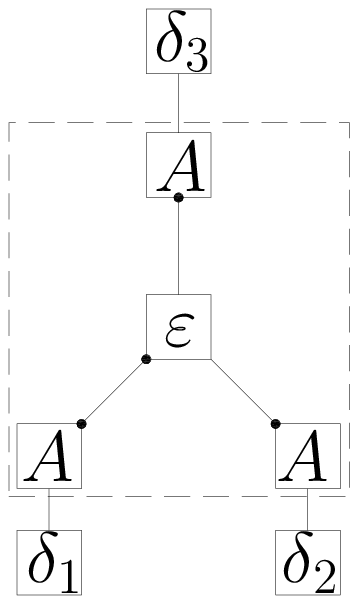}}
\end{picture}
\end{figure}
\hspace{-.42cm} which proves
\[
(a_1 \times a_2)\cdot a_3 = (a_2 \times a_3)\cdot a_1 = (a_3 \times a_1)\cdot a_2 = \det(A),
\]
where $A$ is the $3 \times 3$ matrix $(a_1 \ a_2 \ a_3)$ and for each $i \in \{1,2,3 \}$, $\delta_{i}(x) = \delta(x-i)$ for all $x \in \{1,2,3\}$.


\subsection{Pfaffian}

Given a $2n \times 2n$ skew-symmetric matrix $A$, the \emph{Pfaffian} of $A$, denoted $\Pf(A)$, is defined as
\[
\Pf(A) = \frac{1}{2^n n!} \sum_{\sigma \in S_{2n}} \sgn(\sigma) \prod_{i = 1}^{n} A\big(\sigma(2i-1), \sigma(2i)\big).
\]
Before we proceed, we need the following lemma.
\begin{Lemma}
  For any positive integer $n$, let $\tau \in S_{2n}$ be such that $\tau(2k-1)
  = k$ and $\tau(2k) = 2n-(k-1)$, $1 \leq k \leq n$. Then $\tau$ is an even
  permutation.
 \label{lemma:perm}
\end{Lemma}
\begin{proof}
For brevity, we use the tuple $(i_1, \ldots, i_{2n})$ 
to refer to the permutation 
$\left(	\begin{array}{ccc}	1&\cdots&{2n}\cr i_1&\cdots&i_{2n}		\end{array}		\right)$. Now, given the identity permutation $(1, \ldots, n, n+1, \ldots, 2n)$, it is possible to reflect the second half of the permutation so that it becomes $(1, \ldots, n, 2n, \ldots, n+1)$; this clearly can be done using $\left\lfloor \frac{n}{2} \right\rfloor$ swaps. Next, we interleave the two halves of the $2n$-tuple by performing the following $n-1$ steps: Starting with $k = 1$ and ending with $k = n-1$, at step $k$ imagine a window that starts just before the $2k$th element of the tuple and ends just after the $n+k$th element (i.e., it covers $n-k+1$ elements). For the elements of the tuple covered by the window, we perform a cyclic-shift by one position to the right. It is easy to check that after the $(n-1)$th step we arrive to our desired permutation. Clearly, a one-position cyclic shift on $m$ elements can be accomplished using $m-1$ swaps. Hence, interleaving the tuple requires $\sum_{m = 2}^{n} (m-1) = \frac{n(n-1)}{2}$ swaps. Therefore, in total our permutation can be written in terms of $\left\lfloor \frac{n}{2} \right\rfloor + \frac{n(n-1)}{2}$ swaps and the claim follows by noting that this number is even for any positive integer $n$.
\end{proof}

The following proposition affirmatively proves Peterson's conjecture on the Pfaffian \cite{Peterson:Cayley-Hamilton} \cite{Peterson:Unshakling}.
\begin{Prop}
Let $A$ be a $2n \times 2n$ skew-symmetric matrix and let $\G$ be as in
Fig.~\ref{fig:prop_Pfaffian}. Then $Z_{\G} = n! \cdot 2^n \cdot \emph{\Pf(A)}$.
\label{prop:Pfaffian}
\begin{figure}[H]
  \vskip0.25cm
	\centering\setlength{\unitlength}{0.5cm}	
	\begin{picture}(10,3)(0,0)
		\put(-1,0){\includegraphics[scale = .6]{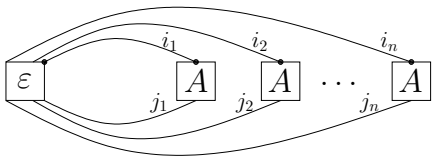}	}
	\end{picture}
	\caption{: The NFG $\G$ from Proposition~\ref{prop:Pfaffian}.}
	\label{fig:prop_Pfaffian}
\end{figure}
\end{Prop}

\vspace{-.1cm}
\begin{proof}
We have
\begin{eqnarray*}
Z_{\G} &=& \!\!\!\!\!\!\!\!\!\!\!\!
\sum_{(i_1, \ldots, i_n, j_{n}, \ldots, j_1) \in N^{2n}} \!\!\!\!\!\!\!\!\!\!
\varepsilon(i_1, \ldots, i_n, j_n, \ldots, j_1)  \prod_{k = 1}^{n} A(i_k, j_k) \\ 
&=& \sum_{\sigma \in S_{2n}} \sgn(\sigma) \prod_{k = 1}^{n} A\big(\sigma(k), \sigma\left(2n-(k-1)\right)\big).
\end{eqnarray*}
Now let $\sigma = \left( \begin{array}{ccccccc} 1&2&\ldots&n&n+1&\ldots&2n\cr
    i_1&i_2&\ldots&i_n&j_{n}&\ldots&j_1 \end{array} \right)$ and $\sigma' =
\left( \begin{array}{ccccccc} 1&2&3&4&\ldots&2n-1&2n\cr
    i_1&j_1&i_2&j_2&\ldots&i_n&j_n \end{array} \right)$ and note that $\sigma'
= \sigma \circ \tau$ where $\tau$ is as in Lemma \ref{lemma:perm}. Hence,
$\sgn(\sigma') = \sgn(\sigma) \cdot \sgn(\tau) = \sgn(\sigma)$, where the last
equality follows from Lemma \ref{lemma:perm}. Further, note that for all $k
\in \{1, \ldots, n\}$ we have $\sigma(k) = i_k = \sigma'(2k-1)$ and $\sigma(2n
- (k-1)) = j_k = \sigma'(2k)$. Finally, it is clear that as $\sigma$ runs over
$S_{2n}$, $\sigma'$ runs over $S_{2n}$. Hence, $Z_{\G} = \sum_{\sigma' \in
  S_{2n}} \sgn(\sigma') \prod_{k = 1}^{n} A\big(\sigma'(2k-1),
\sigma'(2k)\big)$, and the claim follows by the definition of the Pfaffian.
\end{proof}

\section{Conclusion}
\label{sec:Conc}
 This paper presents NFGs as an intuitive tool that facilitates proofs and enhances understanding of some notions and identities in linear algebra. We feel that this new look at NFGs, together with the related work \cite{Forney2011:PartitionFunction},  may result in a better understanding of NFGs as a model of codes and provide deeper understanding of the message-passing algorithms. Further, it is our hope that  the use of NFGs as an algebraic tool may open doors for developing new algorithmic tools that 
 are useful in a broader range of disciplines.

\section*{Acknowledgment}
The authors wish to thank David Forney for several discussions on the subject.

\bibliography{C:/Research_Latex/bibliographys/Holographic_revised}

\end{document}